\documentclass[a4paper,oneside,reqno]{amsart}

\usepackage{amssymb,cite}

\numberwithin{equation}{section}

\newcommand{\RR}{\mathbb{R}}
\newcommand{\CC}{\mathbb{C}}
\newcommand{\NN}{\mathbb{N}}
\newcommand{\R}{\mathbb{R}}

\newcommand{\cG}{\mathcal{G}}
\newcommand{\cO}{\mathcal{O}}
\newcommand{\cL}{\mathcal{L}}
\newcommand{\cD}{\mathcal{D}}

\newcommand{\dist}{{\rm dist}}

\newtheorem{theorem}{Theorem}[section]
\newtheorem{lemma}[theorem]{Lemma}
\newtheorem{prop}[theorem]{Proposition}

\theoremstyle{remark}

\newtheorem{remark}[theorem]{Remark}

\DeclareMathOperator{\supp}{supp}
\DeclareMathOperator{\vol}{vol}
\DeclareMathOperator{\area}{area}

\begin{document}

\title[]{On eigenvalue asymptotics for strong $\delta$-interactions supported by surfaces with boundaries}

\author{Jaroslav Dittrich}
\address{Department of Theoretical Physics, Nuclear Physics Institute, Czech Academy of Sciences, 25068 \v Re\v z near Prague, Czechia, and
Doppler Institute for Mathematical Physics and Applied Mathematics, Czech Technical University, B\v rehov\'a 7, 11519 Prague, Czechia}
\email{dittrich@ujf.cas.cz}

\author{Pavel Exner}
\address{Department of Theoretical Physics, Nuclear Physics Institute, Czech Academy of Sciences, 25068 \v Re\v z near Prague, Czechia, and
Doppler Institute for Mathematical Physics and Applied Mathematics, Czech Technical University, B\v rehov\'a 7, 11519 Prague, Czechia}
\email{exner@ujf.cas.cz}

\author{Christian K\"uhn}
\address{Technische Universit\"at Graz, Institut f\"ur Numerische Mathematik, Steyrergasse 30, 8010 Graz, Austria}
\email{kuehn@tugraz.at}

\author{Konstantin Pankrashkin}
\address{Laboratoire de math\'ematiques d'Orsay (CNRS UMR 8628), Universit\'e Paris-Sud, B\^atiment 425, 91400 Orsay, France}
\email{konstantin.pankrashkin@math.u-psud.fr}

\thanks{Jaroslav Dittrich, Pavel Exner and Christian K\"uhn gratefully acknowledge a partial support of the Czech Science Foundation (GA\v{C}R), project 14-06818S, and the Austrian-Czech Republic cooperation grant CZ01/2013.
Konstantin Pankrashkin was partially supported by ANR NOSEVOL (ANR 2011 BS01019 01) and GDR CNRS 2279 DynQua. Christian K\"uhn gratefully acknowledges financial support by the Austrian Science Fund (FWF), project P 25162-N26.
}


\begin{abstract}
Let $S\subset\mathbb{R}^3$ be a $C^4$-smooth relatively compact orientable surface with a sufficiently regular boundary.
For $\beta\in\RR_+$, let $E_j(\beta)$ denote the $j$th negative eigenvalue of the operator associated with the quadratic
form
\[
H^1(\RR^3)\ni u\mapsto \iiint_{\mathbb{R}^3} |\nabla u|^2dx -\beta \iint_S |u|^2d\sigma,
\]
where $\sigma$ is the two-dimensional Hausdorff measure on $S$. We show that for each fixed $j$ one has the asymptotic expansion
\[
E_j(\beta)=-\dfrac{\beta^2}{4}+\mu^D_j+ o(1) \;\text{ as }\; \beta\to+\infty\,,
\]
where $\mu_j^D$ is the $j$th eigenvalue of the operator $-\Delta_S+K-M^2$ on $L^2(S)$, in which $K$ and $M$ are  the Gauss and mean curvatures, respectively, and $-\Delta_S$ is the Laplace-Beltrami
operator with the Dirichlet condition at the boundary of $S$. If, in addition, the boundary of $S$ is $C^2$-smooth, then the remainder estimate can be improved to
${\mathcal O}(\beta^{-1}\log\beta)$.
\end{abstract}

\maketitle

\section{Introduction}

Schr\"{o}dinger operators with singular interaction are often used to produce solvable models of various sorts and to study relations between spectral properties and the geometry of the interaction support. While the case where the latter is a discrete point set was an object of interest for more than three decades \cite{AGHH}, relatively less attention has been paid to contact interactions supported by hypersurfaces in Euclidean spaces. The first general and mathematically rigorous analysis of such Schr\"{o}dinger operators with interaction support of codimension one was presented in \cite{BEKS}.
However, it was only in 2001, when a wave of interest to these operators started, its initial point being the paper \cite{EI}, motivated by the fact
that they can model semiconductor graph-like structures which, in contrast to the usual quantum graph theory \cite{BK}, make it possible to take quantum tunneling into account. At the same time interesting mathematical questions appeared. For a summary of the work done in the following few
years we refer the reader to the review paper \cite{Ex08} and the recent papers~\cite{BEL,BLL13,Loto, Mant}.

A question of a particular interest concerned strong coupling situations, i.e.\ the asymptotic behavior of eigenvalues when the coupling constant $\beta$ of an attractive singular interaction becomes large. A technique was devised in \cite{EY2, EY} which allowed to treat this problem for smooth curves $\Gamma$ in $\RR^2$  without endpoints by a combination of bracketing technique and a use of locally orthogonal coordinates in the vicinity of the support. In this way an asymptotic expansion was obtained in which a universal divergent first term was followed by an eigenvalue of a one-dimensional Schr\"odinger operator with the potential $-\gamma^2/4$ where $\gamma$ was the curvature of $\Gamma$. An analogous result was later derived for a closed smooth surface $S$ in $\RR^3$; in that case the comparison operator was two-dimensional and the curvature induced potential was $K-M^2$ where $K$ and $M$ were the Gauss and mean curvatures of $S$, respectively. We note in passing that similar asymptotic
formul{\ae} can be derived also for more
singular interactions such as those supported by curves in $\RR^3$, i.e.\ manifolds of codimension two, see~\cite{EK},
or $\delta'$-interactions supported by planar curves~\cite{exj}.

The method of \cite{EY} did not work, however, for manifolds with a boundary, or more exactly, it did work but it provided estimates too rough to characterize individual eigenvalues. The reason was that the used upper and lower bounds in this case yielded comparison operators with the same symbol but different boundary conditions, Dirichlet and Neumann. The conjecture was that the asymptotic formula still holds with the \emph{Dirichlet} comparison operator. This claim was indeed found valid: in \cite{EP} it was proved for non-closed smooth finite curves $\Gamma$ in $\RR^2$ by using bracketing in an extended neighborhood of $\Gamma$ in combination with a convolution-type expression of the eigenfunction. The aim of the present paper is to solve the analogous problem for finite smooth surfaces in $\RR^3$ with a boundary.

Let us describe the main results more rigorously.
Let $S\subset \RR^3$ be a $C^4$-smooth relatively compact orientable surface with a compact Lipschitz boundary $\partial S$.
In addition, we suppose that $S$ can be extended through the boundary, i.e.\ that there exists
a larger $C^4$-smooth surface $S_2$ such that $\overline{S}\subset{S_2}$, where $\overline{S}$ means the closure of $S$.
We consider the quadratic form
\[
h(u,u) = \iiint_{\RR^3}|\nabla u|^2dx -\beta \iint_S |u|^2\,d\sigma, \quad \cD(h)=H^1(\RR^3),
\]
where $\sigma$ is the two-dimensional Hausdorff measure on $S$. One can easily check (see \cite{BEKS}) that this form is closed in $L^2(\R^3)$
and semibounded from below, and hence it defines a unique self-adjoint operator $H$ in $L^2(\RR^3)$, semibounded from below.
In a suitable weak sense, the operator $H$ can be interpreted as the Laplacian with the boundary condition $[\partial_\nu u]=\beta u$ on $S$, where
$[\partial_\nu u]$ is the jump of the normal derivative on $S$, see e.g. \cite{BEKS} or \cite{Mant}.
Using the standard machineries one shows that the essential spectrum $\sigma_\mathrm{ess}(H)$ of $H$ is $[0,\infty)$
and that there is a finite number of negative eigenvalues $E_1(\beta)<E_2(\beta)\le\dots$, cf. \cite{BEKS}.

The embedding $S\subset \RR^3$ gives rise to a metric tensor $(g_{ab})$ on $S$ and to the contravariant tensor
$(g^{ab}):=(g_{ab})^{-1}$, and for the Hausdorff measure $\sigma$ we have $d\sigma(s)=\sqrt g ds$ with $g:=\det(g_{ab})$.
We will deal with the operator $L^D_S:=-\Delta_S+K-M^2$ on $S$ with the Dirichlet boundary condition at the boundary of~$S$,
with  $K$ and $M$ being respectively the Gauss and mean curvature on $S$.
More precisely, the operator $L^D_S$ is defined as the unique self-adjoint operator acting in $L^2(S):=L^2(S,\sigma)$
generated by the quadratic form
\begin{equation}
       \label{eq-lsd}
H^1_0(S)\ni u\mapsto \langle \partial_j u,g^{jk} \partial_k u\rangle_{L^2(S)}
+\langle u, (K-M^2)u\rangle_{L^2(S)}\,,
\end{equation}
where the Einstein convention for the indices is used. Due to the relative
compactness of $S$ in $S_2$, both $K$ and $M$ are bounded on $S$, and
the operator $L^D_S$ is semibounded from below and has a compact resolvent.
We denote by $\mu^D_j$, $j\in\NN$, its eigenvalues enumerated in the usual way.

The results of the paper can be summarized as follows:
\begin{theorem}\label{thmain}
Let $S$ be as above (i.e.\ relatively compact with a Lipschitz boundary and extendable through the boundary), then
for each fixed $j\in\NN$ one has
\begin{equation}
     \label{eq-eb1}
E_j(\beta)=-\beta^2/4 + \mu_j^D+o(1)\text{ as } \beta\to+\infty.
\end{equation}
If, in addition, $S$ has a $C^2$-boundary, then the remainder estimate can be replaced by $\cO(\beta^{-1}\log\beta)$.
\end{theorem}

The two remainder estimates are obtained by different methods. The proof of \eqref{eq-eb1} relies on the monotone convergence of non-densely defined quadratic forms  \cite{W}, which is a new tool in comparison to the previous papers on $\delta$-interactions.
The proof is contained in Sections~\ref{s: first}--\ref{s: weakrem}.
We note that the form convergence we use appears to be similar to that appearing in analysis of the strong coupling limit of operators $-\Delta+\lambda\chi_\Omega$, where $\chi_\Omega$ is the indicator function of a set $\Omega$, see~\cite{sim}. It is known that the regularity of $\Omega$ in this problem  plays an important role in calculating the convergence rate  of such operators \cite{vbr, bd}.
While the asymptotics~\eqref{eq-eb1} gives the expected result, we have only a weak estimate of the error term.
In order to obtain the same remainder estimate $\cO(\beta^{-1}\log\beta)$ as for closed surfaces,
we are adapting to the present situation the technique used in \cite{EP},
which is done in the second part of the paper, from Section~\ref{s: kernel} on, and which appears to be rather
technically involved compared to the two-dimensional case.

\section{First steps}
\label{s: first}

\subsection{Analysis on thin neighborhoods}
\label{s: thin}

We remark that due to the properties of $S$ we can pick an ``intermediate'' surface $S_1$ which is relatively compact, $C^4$, with a Lipschitz boundary,
and such that $\overline S\subset S_1 \subset \overline{S_1}\subset S_2$.
Furthermore, let $S_2\ni s\mapsto \nu(s)$ be a smooth unit normal on $S_2$.

It is a well-known fact of the differential geometry that for a sufficiently small $a>0$, the map
\begin{equation}
        \label{eq-sff}
S_1\times(-a,a)\ni (s,t)\mapsto F (s,t)=s+t \nu(s)\in \RR^3
\end{equation}
is a diffeomorphism between $S_1\times(-a,a)$ and its image and can be continued to the boundary.
We  introduce the following sets (we omit the dependence on $a$):
\[
U_1:=F \big(S_1\times(-a,a)\big), \quad
U:= F \big(S\times(-a,a)\big).
\]
and consider the following two quadratic forms,
\begin{align}
h^N_1(u,u)&:=\iiint_{U_1} |\nabla u|^2dx-\beta \iint_S |u|^2 \,d\sigma\,, & \cD(h^N_1)&=H^1(U_1)\,,\\
\label{HDdef}
h^D (u,u)&:=\iiint_{U} |\nabla u|^2dx-\beta \iint_S |u|^2 \,d\sigma\,, & \cD(h^D)&=H^1_0(U)\,,
\end{align}
and denote by $H^N_1$ and $H^D$ the associated self-adjoint operators acting respectively in $L^2(U_1)$ and $L^2(U)$. In the sense of forms we have then the inequalities $H^N_1\oplus (-\Delta^N_1)\le H \le H^D\oplus (-\Delta^D)$, where $-\Delta^N_1$ is the Neumann Laplacian in $L^2(\RR^3\setminus \Bar U_1)$ and $-\Delta^D$ is the Dirichlet Laplacian in $L^2(\RR^3\setminus \Bar U)$.
As both $-\Delta^N_1$ and $-\Delta^D$ are non-negative, to assess the negative spectrum it is sufficient to compare the negative eigenvalues of $H$ with those of $H^N_1$ and $H^D$ which have both compact resolvents.

Through the text, for a semibounded from below self-adjoint operator $A$ we denote by $\Lambda_j(A)$ its $j$th eigenvalue (provided it exists).
Then the above consideration gives the inequalites
\[
\Lambda_j(H^N_1)\le E_j(\beta) \equiv \Lambda_j(H) \le \Lambda_j(H^D)\,,
\]
valid (at least) for the indices $j$ for which the right-hand side is negative.
Using the above diffeomorphism $F$ we introduce unitary tranformations $\phi$ and $\phi_1$
as follows:
\[
\phi_1: L^2(U_1) \to L^2(S_1\times (-a,a), d\sigma\,dt) \,,\quad
(\phi_1 f)(s,t):=\sqrt{\big(1+k_1(s) t\big)\big(1+k_2(s) t\big)}\,f\big(F(s,t)\big)
\]
with $k_1$ and $k_2$ being the principal curvatures, and $\phi$ is defined analogously with $U_1$ and $S_1$ replaced by $U$ and $S$.
A standard computation, see e.g.~\cite{DEK,ex},
shows that, in the sense of forms, one can estimate $B^N_1\le \phi_1 H^N_1 \phi^{-1}_1$ and $\phi H^D \phi^{-1}\le B^D$, where $B^N_1$ and $B^{D}$ are the self-adjoint operators, acting respectively in $L^2(\Sigma_1):=L^2(\Sigma_1,d\sigma\,dt)$
 and $L^2(\Sigma):=L^2(\Sigma, d\sigma\,dt)$, where
\[
\Sigma_1:=S_1\times(-a,a), \quad \Sigma:=S\times(-a,a),
\]
 associated respectively with the forms $b^N_1$ and $b^D$,
\[
\begin{aligned}
b^N_1 (u,u)=\,&(1-Ca)\Big(\langle \partial_j u,g^{jk}\partial_k u\rangle_{L^2(\Sigma_1)}
+ \langle u,(K-M^2)u\rangle_{L^2(\Sigma_1)}\Big) + \|\partial_t u\|^2_{L^2(\Sigma_1)}\\
&-\beta\iint_S |u(s,0)|^2d\sigma -
Ca \|u\|^2_{L^2(\Sigma_1)} -C\iint_{S_1} \Big(
\big|u(s,-a)\big|^2+\big|u(s,a)\big|^2
\Big)d\sigma\,,
\\
\cD(b^N_1)&=H^1(\Sigma_1),
\\
b^D (u,u)=\,&(1+Ca)\Big(\langle \partial_j u,g^{jk}\partial_k u\rangle_{L^2(\Sigma)}
+\langle u,(K-M^2)u\rangle_{L^2(\Sigma)} \Big) + \|\partial_t u\|^2_{L^2(\Sigma)}\\
& -\beta\iint_S |u(s,0)|^2d\sigma + Ca \|u\|^2_{L^2(\Sigma)}\,,\\
\cD(b^D)&=H^1_0(\Sigma)\,,
\end{aligned}
\]
where $C>0$ is  independent of $a$ and $\beta$.
We remark again that  $B^N_1$ and $B^D$ have compact resolvents, and we arrive at the inequalities
\begin{equation}
      \label{eq-lll}
\Lambda_j(B^N_1)\le \Lambda_j(H)\le \Lambda_j(B^D) \quad \text{ if } \quad \Lambda_j(B^D)<0.
\end{equation}
The sought asymptotic expansions will arise from estimating the both bounds of~\eqref{eq-lll} which we will do
in the subsequent sections.

\subsection{Separation of variables and upper bound}
\label{s: majoration}

The upper bound of $E_j(\beta)$ in terms of $\mu_j^D$ is relatively easy, and the
right-hand side of \eqref{eq-lll}  was in fact already estimated in~\cite{ex}.
We repeat the construction here for the sake of completeness:

\begin{lemma}\label{lem1}
Assume that $S$ has a compact Lipschitz boundary, then
for any fixed $j\in\NN$ there holds
\[
E_j(\beta)\le-\dfrac{1}{4}\,\beta^2+\mu^D_j + \cO\Big(\dfrac{\log\beta}{\beta}\Big) \text{ as } \beta\to+\infty.
\]
\end{lemma}

In order to prove Lemma~\ref{lem1} we remark that one can represent $B^D:=\widetilde L^D_S \otimes 1+ 1\otimes T^D$, where $\widetilde L^D_S=(1+Ca)L^D_S+Ca$ and $T^D$ is the self-adjoint operator in $L^2(-a,a)$ associated with the quadratic form
\[
H^1_0(-a,a)\ni v\mapsto \int_{-a}^a\big|v'(t)\big|^2dt-\beta \big|v(0)\big|^2.
\]
By \cite[Proposition 2.4]{EY}, for $\beta a> 8/3$ the operator $T^D$ has a unique negative eigenvalue,
and
\[
-\dfrac{\beta^2}{4}\le\Lambda_1(T^D)\le-\dfrac{\beta^2}{4}+2\beta^2 e^{-\beta a/2}.
\]
At the same time we have $\Lambda_j(\widetilde L^D_S)=(1+Ca)\mu^D_j+C a$, and $\mu_j^D$ do not depend on $a$.
Therefore, if $a$ is small and both $\beta$ and $\beta a$ are large, one has
\begin{equation}
     \label{eq-bdd}
\Lambda_j(B^D) \le -\dfrac{\beta^2}{4}+2\beta^2 e^{-\beta a/2} + \mu_j^D + Ca (1+\mu^D_j)
\end{equation}
for all $j$ with $(1+Ca)\mu^D_j+C a \le -\Lambda_1(T^D)$.
Therefore, if $j$ is fixed, $\beta$ is large and
\begin{equation}
  \label{eq-aa}
a:=\xi \beta^{-1}\log\beta, \quad \xi\ge 6,
\end{equation}
we have
\begin{equation}
     \label{eq-min}
\Lambda_j(B^D)\le-\dfrac{1}{4}\,\beta^2+\mu^D_j + \cO\Big(\dfrac{\log\beta}{\beta}\Big),
\end{equation}
and Lemma~\ref{lem1} is obtained by the substitution into \eqref{eq-lll}.

\begin{remark}\label{rem3}
For a later use, we remark that a similar approach can be applied to other related operators.
Namely, for sufficiently small $a>0$ (supposed to be less than one in order to avoid a notation conflict), consider the surface
\[
S_a=\{z\in S_1:\, \dist_{S_1}(z,S)<a\},
\]
where $\dist_{S_1}$ is the distance measured along the geodesics of $S_1$, and denote
\begin{equation}
 \label{eq-xia}
\Xi_a= F\big(S_a\times (-a,a)\big).
\end{equation}
Denote by $H^D_a$ the self-adjoint operator acting in $L^2(\Xi_a)$
generated by the quadratic form
\[
h^D_a(u,u)=\iiint_{\Xi_a} |\nabla u|^2 ds -\beta \iint_{S_a} |u|^2 d\sigma,
\quad \cD (h^D_a)=H^1_0(\Xi_a),
\]
then using the same computations one obtains, with $a$ chosen as \eqref{eq-aa},
\begin{equation}
    \label{eq-hda}
\Lambda_j(H^D_a)= -\dfrac{1}{4}\,\beta^2+\mu^D_j(a) + \cO\Big(\dfrac{\log\beta}{\beta}\Big),
\end{equation}
where $\mu^D_j(a)$ is the $j$th Dirichlet eigenvalue of $-\Delta_S+K-M^2$
on $S_a$.
\end{remark}


\section{Eigenvalue asymptotics with a weak remainder estimate}
\label{s: weakrem}

In the present section we are going to prove the first part of Theorem~\ref{thmain}:
\begin{prop}\label{thm1}
Assume that $S$ has a compact Lipschitz boundary, then
for any fixed $j\in\NN$ there holds $\,E_j(\beta)=-\beta^2/4 +\mu_j^D+o(1)$ as $\,\beta\to+\infty$.
\end{prop}

We recall that in this claim we have no control over the remainder. On the other hand, we impose here quite weak assumptions concerning the regularity of the boundary of $S$, and our approach is quite robust; we expect that it can be adapted to similar problems like the strongly attractive $\delta'$-interactions with minimal effort. In view of Lemma~\ref{lem1} we just need to establish a suitable lower bound for $E_j(\beta)$, which is done
in the rest of the section.

\subsection{A one-dimensional operator}

Let us keep the choice made in~\eqref{eq-aa} for $a$ and denote by $T^N$ the self-adjoint operator in $L^2(-a,a)$ associated with the quadratic form
\[
t^N(v,v)=\int_{-a}^a|v'|^2dt-\beta \big|v(0)\big|^2 -C \Big(
\big|v(-a)\big|^2+\big|v(a)\big|^2\Big)\,, \quad \cD(t^N)=H^1(-a,a)\,.
\]
As shown in~\cite[Proposition~2.5]{EY}, for $\beta\to+\infty$ we have
\begin{equation}
      \label{eq-pi}
\Lambda_1(T^N)=-\dfrac{1}{4}\,\beta^2+ \cO\Big(\dfrac 1 \beta\Big)\,,
\;\;
\Lambda_2(T^N)\ge 0\,.
\end{equation}
Let $\varphi_j$ denote normalized eigenfunctions corresponding to the eigenvalues $\Lambda_j(T^N)$ and let $P_j$ be the orthogonal projectors on $\CC\varphi_j$ in $L^2(-a,a)$; we recall that all the eigenvalues of $T^N$ are simple. In virtue of the spectral theorem for self-adjoint operators
we have the inequality
\begin{align*}
t^N(v,v)&\geq  \Lambda_1(T^N) \|P_1v\|^2+ \Lambda_2(T^N) \|(1-P_1)v\|^2 \\
&= \Lambda_1(T^N) \|v\|^2+\big(\Lambda_2(T^N)-\Lambda_1(T^N)\big) \|(1-P_1)v\|^2\,, \quad v\in \cD(t^N)\,,
\end{align*}
which can be rewritten as
\begin{multline}
         \label{eq-tn}
\int_{-a}^a |v'|^2dt-\beta |v(0)|^2- C \Big(
\big|v(-a)\big|^2+\big|v(a)\big|^2\Big)-\Lambda_1 (T^N) \int_{-a}^a|v|^2\,dt\\
\ge
\big(\Lambda_2(T^N)-\Lambda_1(T^N)\big)
\int_{-a}^a \big|(1-P_1)v\big|^2 dt\,, \quad  v\in H^1(-a,a)\,.
\end{multline}

\subsection{Minoration of the quadratic form}

We denote
\[
\Omega:=S_1\setminus \overline{\mathstrut S}\,, \quad \Xi:=\Omega\times(-a,a)\equiv \Sigma_1\setminus\overline{\mathstrut\Sigma}\,.
\]
By regrouping the terms in the expression for $b^N_1$
we obtain
\begin{multline}
b^N_1 (u,u)-\Lambda_1(T^N)\|u\|^2_{L^2(\Sigma_1)}\\
\begin{aligned}
=&\,(1-Ca)\langle \partial_j u,g^{jk}\partial_k u\rangle_{L^2(\Sigma_1)} + \langle u,(K-M^2)u\rangle_{L^2(\Sigma_1)}
\\
&+ \|\partial_t u\|^2_{L^2(\Sigma)}-\beta\iint_S |u(s,0)|^2d\sigma \\
&\qquad\qquad-C\iint_S \Big(
\big|u(s,-a)\big|^2+\big|u(s,a)\big|^2\Big)d\sigma-\Lambda_1(T^N)\|u\|^2_{L^2(\Sigma)}
\\
& +  \|\partial_t u\|^2_{L^2(\Xi)} -C\iint_{\Omega} \Big(\big|u(s,-a)\big|^2+\big|u(s,a)\big|^2\Big)d\sigma-\Lambda_1(T^N)\|u\|^2_{L^2(\Xi)}
\\
&- Ca \langle u,(K-M^2+1)u\rangle_{L^2(\Sigma_1)}.
\end{aligned}
 \label{eq-bn}
\end{multline}
To simplify some terms on the right-hand side  we use the identification $L^2(\Sigma_1)\equiv L^2(S_1)\otimes L^2(-a,a)$. Consider the orthogonal projector $\Pi_1:=1\otimes P_1$ in $L^2(\Sigma_1)$. Using first \eqref{eq-tn} and then the asymptotics~\eqref{eq-pi}, for any $u\in H^1(\Sigma)$ we obtain, as $\beta$ is large,
\begin{multline}
\|\partial_t u\|^2_{L^2(\Sigma)}-\beta\iint_S \big|u(s,0)\big|^2d\sigma
-C\iint_S \Big(\big|u(s,-a)\big|^2+\big|u(s,a)\big|^2 \Big)d\sigma  -\Lambda_1 (T^N) \|u\|^2_{L^2(\Sigma)} \\
\begin{aligned}
&\ge \big(\Lambda_2(T^N)-\Lambda_1(T^N)\big)\|(1-\Pi_1)u\|^2_{L^2(\Sigma)}\\
&\ge \dfrac{\beta^2}{5}\|(1-\Pi_1)u\|^2_{L^2(\Sigma)}\,. \label{est1}
\end{aligned}
\end{multline}
Furthermore, using the Sobolev inequality~\cite[Lemma~8]{kuch},
\[
\big|v(-a)\big|^2+\big|v(a)\big|^2 \le 4a \|v'\|^2_{L^2(-a,a)}+2a^{-1}\|v\|^2_{L^2(-a,a)}
\;\text{ for any } v\in H^1(-a,a)\,,
\]
and then the asymptotics~\eqref{eq-pi}, which gives
\begin{multline}
\|\partial_t u\|^2_{L^2(\Xi)}-
C\iint_{\Omega} \Big(
\big|u(s,-a)\big|^2+\big|u(s,a)\big|^2
\Big)d\sigma-\Lambda_1(T^N)\|u\|^2_{L^2(\Xi)}\\
\begin{aligned}
&\ge (1-4Ca)\|\partial_t u\|^2_{L^2(\Xi)}-2Ca^{-1}\|u\|^2_{L^2(\Xi)} -\Lambda_1(T^N)\|u\|^2_{L^2(\Xi)}
\\
&\ge \Big(-\Lambda_1(T^N)- \dfrac{2 C\beta}{\xi\log \beta}\Big)\|u\|^2_{L^2(\Xi)}
\\
&\ge \dfrac{\beta^2}{5}\|u\|^2_{L^2(\Xi)}
\\
&=\dfrac{\beta^2}{5}\|\Pi_1 u\|^2_{L^2(\Xi)}+\dfrac{\beta^2}{5}\|(1-\Pi_1)u\|^2_{L^2(\Xi)}\,.
\end{aligned}
   \label{est2}
\end{multline}
Finally, we have
\begin{equation}
\label{eq-ee}
\big| \langle u,(K-M^2+1)u\rangle_{L^2(\Sigma_1)}\big|\le E \|u\|^2_{L^2(\Sigma_1)}
\end{equation}
with $E:=\|K-M^2\|_{L^\infty(S_1)}+2$. Substituting now \eqref{est1}, \eqref{est2} and \eqref{eq-ee} into~\eqref{eq-bn} we conclude that for large $\beta$ one has
\begin{equation}
    \label{eq-bc}
b^N_1-\Lambda_1(T^N)\ge c_\beta\,,
\end{equation}
where $c_\beta$ is the quadratic form given by
\begin{multline*}
\hspace{-.5em} c_\beta(u,u):=(1-Ca)\langle \partial_j u,g^{jk}\partial_k u\rangle_{L^2(\Sigma_1)}
+ \langle u,(K-M^2)u\rangle_{L^2(\Sigma_1)}-Ca E \|u\|^2_{L^2(\Sigma_1)}\\
+\dfrac{\beta^2}{5}\Big(
\|(1-\Pi_1)u\|^2_{L^2(\Sigma_1)}+ \|\Pi_1 u\|^2_{L^2(\Xi)}\Big)\,, \quad
\cD(c_\beta)=H^1(S_1)\otimes L^2(-a,a)\,.
\end{multline*}
Denote by $C_\beta,\: \beta>0,$ the self-adjoint operators associated with $c_\beta$. Since our argument involves
the min-max principle we have to pay attention to the essential spectrum, noting that these operators have no longer compact resolvents. To estimate the essential spectrum threshold, we simply drop the last non-negative summand and write $c_\beta\ge c'_\beta$ with
\begin{multline*}
\hspace{-.5em} c'_\beta(u,u):=
(1-Ca)\langle \partial_j u,g^{jk}\partial_k u\rangle_{L^2(\Sigma_1)} + \langle u,(K-M^2)u\rangle_{L^2(\Sigma_1)}\\
-Ca E \|u\|^2_{L^2(\Sigma_1)}+\dfrac{\beta^2}{5}\|(1-\Pi_1)u\|^2_{L^2(\Sigma_1)}\,,
\quad \cD(c'_\beta)=H^1(S_1)\otimes L^2(-a,a)\,,
\end{multline*}
which means that the self-adjoint operators $C'_\beta$ associated with $c'_\beta$ can be written as
\[
C'_\beta=L^N_S\otimes 1+  1\otimes \dfrac{\beta^2}{5}(1-P_1)\,,
\;\;
L^N_S=(1-Ca)(-\Delta^N_S)+K-M^2-Ca E\,,
\]
the operator $-\Delta^N_S$ being the Neumann Laplace-Beltrami operator in $L^2(S_1)$ associated with the form $H^1(S_1)\ni u\mapsto \langle \partial_j u,g^{jk}\partial_k u\rangle_{L^2(S_1)}$. We note that the operators
$L^N_S$ have compact resolvents and they are uniformly semibounded, $L^N_S\ge -E$ for all sufficiently large $\beta$, and their $j$th eigenvalues behave, as $j$ is fixed,
as $\Lambda_j \big(L^N_S\big) = \Lambda_j(-\Delta^N_S+K-M^2) +\cO(a)$. On the other hand, the spectrum of the transverse part $1-P_1$ consists of a simple eigenvalue zero and the infinitely degenerate eigenvalue $1$, which gives
\[
\inf\sigma_\mathrm{ess} (C'_\beta)\ge
\dfrac{\beta^2}{5}+\inf\sigma (L^N_S)
\ge \dfrac{\beta^2}{5}-E \to +\infty
\;\;\text{ as } \beta\to +\infty\,.
\]
As $C'_\beta\le C_\beta$, this result means at the same time that
\begin{equation}
     \label{eq-cess}
\inf\sigma_\mathrm{ess} (C_\beta) \to +\infty
\;\;\text{ for } \beta\to +\infty.
\end{equation}
In view of \eqref{eq-bc} and the min-max principle, we then have, for each fixed $j$,
\begin{equation}
  \label{eq-bc1}
\Lambda_j(B^N_1)\ge \Lambda_1(T^N)+\Lambda_j(C_\beta)
=-\dfrac{1}{4}\,\beta^2+\Lambda_j(C_\beta)+\cO\Big(\dfrac{1}{\beta}\Big)\,.
\end{equation}

\subsection{Passing to a common Hilbert space}

In order to deal with a family of forms acting on a fixed Hilbert space we denote $\cG:=L^2(S_1)\otimes \ell^2(\NN)$ and introduce unitary operators $\theta$, $\Theta$, and orthogonal projectors $\kappa_1$, $K_1$ as follows:
\begin{align*}
\theta&:L^2(-a,a)\to \ell^2(\NN), & (\theta f)(j)&:=\langle \varphi_j,f\rangle, \ j\in\NN\,,\\
\Theta&:L^2(S_1)\otimes L^2(-a,a)\to \cG, &  \Theta&:=1\otimes\theta,\\
\kappa_1&: \ell^2(\NN)\to \CC e_1, & \kappa_1 f&:=f(1) e_1, \\
K_1&: L^2(S_1)\otimes\ell^2(\NN) \to L^2(S_1)\otimes\CC e_1, & K_1&:=1\otimes\kappa_1
\end{align*}
where $e_1=(1,0,0,\dots)\in\ell^2(\NN)$. Finally, we introduce the natural identification operator
\[
I:L^2(S_1)\otimes \CC e_1 \mapsto L^2(S_1)\,.
\]
One easily checks that the operators $\widehat C_\beta:=\Theta C_\beta \Theta^*$ are those associated with the quadratic forms
\begin{multline}
      \label{eq-chat}
\widehat c_\beta(u,u)=(1-Ca)\langle \partial_j u,g^{jk}\partial_k u\rangle_\cG + \langle u,(K-M^2)u\rangle_\cG-CaE\|u\|^2_\cG\\
+\dfrac{\beta^2}{5}\Big(
\|(1-K_1)u\|^2_{\cG}+ \|IK_1 u\|^2_{L^2(\Omega)}\Big)\,, \quad
\cD(\widehat c_\beta)=H^1(S_1)\otimes \ell^2(\NN)\,.
\end{multline}
In view of the unitary equivalence between $C_\beta$ and $\widehat C_\beta$, Eqs.~\eqref{eq-cess} and~\eqref{eq-bc1} imply
\begin{gather}
  \label{eq-bc2}
\Lambda_j(B^N_1)\ge -\dfrac{1}{4}\,\beta^2+\Lambda_j(\widehat C_\beta)+\cO\Big(\dfrac{1}{\beta}\Big)\,,\\[.3em]
  \label{eq-bc3}
\inf\sigma_\mathrm{ess} (\widehat C_\beta) \to +\infty \;\;\text{ as } \beta\to +\infty\,.
\end{gather}

\subsection{Convergence of forms}

Recall that we have defined $E:=\|K-M^2\|_{L^\infty(S_1)}+2$. One can pick a $\beta_0$ sufficiently large so  $\widehat c_\beta\ge -(E-1)$ holds for $\beta\ge\beta_0$ and $\widehat c_{\beta_2}\ge \widehat c_{\beta_1}$ for $\beta_2\ge \beta_1\ge \beta_0$, which implies by \cite[Theorem~VI.2.21]{kato} that
\begin{equation}
      \label{eq-ccc}
\big(\widehat C_{\beta_2}+ E\big)^{-1}\le \big(\widehat C_{\beta_1}+ E\big)^{-1}
\le 1
\;\;\text{ as } \beta_2\ge\beta_1\ge \beta_0\,.
\end{equation}
Consider the quadratic form
\[
q(u,u)=\sup_{\beta\ge\beta_0} \widehat c_\beta(u,u)\,, \quad
\cD(q):=\bigg\{u\in \bigcap_{\beta\ge \beta_0} \cD(\widehat c_\beta):
\, \sup_{\beta\ge\beta_0} \widehat c_\beta(u,u)<+\infty
\bigg\}\,.
\]
It is known, see \cite[Theorem~VIII.3.13a]{kato}, that $q$ is closed, and hence it defines a self-adjoint operator $Q\ge-(E-1)$ acting in the Hilbert space $\cL:=\overline{\cD(q)}$, the closure being taken in the topology of $\cG$, and if $\tau:\cG\to\cL$ denotes the orthogonal projection, then
\begin{equation}
   \label{eq-qc2}
(\widehat C_\beta+ E)^{-1}\to \tau^* (Q+E)^{-1}\tau \;\;\text{ strongly as }
\beta\to+\infty\,,
\end{equation}
cf.~\cite[Satz 3.1]{W}. Furthermore, by~\cite[Satz 2.2]{W} we have also
\begin{equation}
   \label{eq-qc1}
\tau^* (Q+E)^{-1}\tau \le (\widehat C_\beta+ E)^{-1} \;\;\text{ for all }\beta\ge\beta_0\,.
\end{equation}
In view of the explicit expression for $\cD(\widehat c_\beta)$ we see that $u\in \cD(q)$ if and only if $u\in H^1(S_1)\otimes \ell^2(\NN)$ such that
$(1-K_1)u=0$ and $\|IK_1u\|_{L^2(\Omega)}=0$. The first condition says that $u$ is of the form
$u=f\otimes e_1$ with $f\in H^1(S_1)$ and $e_1=(1,0,0,\dots)\in\ell^2(\NN)$, and then the second condition tells us that the function $f$ must verify $\|f\|^2_{L^2(\Omega)}=0$. As we supposed that the boundary $\partial S$ is Lipschitz, we have
\[
\Big(f\in H^1(S_1) \text{ and } \|f\|^2_{L^2(\Omega)}=0\Big)
\text{ iff } f\in \widetilde H^1_0(S_1) :=\big\{
f\in H^1(S_1): f|_S\in H^1_0(S), \quad f|_{\Omega}=0
\big\}\,.
\]
Using the unitary operator $J:\widetilde H^1_0(S_1)\to H^1_0(S)$, $Jf=f|_S$, we may write
\begin{gather*}
\cD(q)= \big\{f\otimes e_1: f\in \widetilde H^1_0(S_1)\big\} =
\big\{J^*f\otimes e_1: f\in H^1_0(S)\big\}
\equiv K_1^*I^*J^* H^1_0(S)\,.
\end{gather*}
If $\cD(q)\ni u=J^*f\otimes e_1$ with $f\in H^1_0(S)$, then $f= J I K_1 u$. Substituting into~\eqref{eq-chat} we observe that the last summand equals zero and that the factor $Ca=C \xi\beta^{-1}\log\beta$ vanishes in the limit $\beta\to+\infty$, so we arrive at
\begin{align*}
q(u,u)&= \langle \partial_j u,g^{jk}\partial_k u\rangle_{L^2(\Sigma_1)} + \langle u,(K-M^2)u\rangle_{L^2(\Sigma_1)}
 \\
&=\langle \partial_j u,g^{jk}\partial_k u\rangle_{L^2(\Sigma)} + \langle u,(K-M^2)u\rangle_{L^2(\Sigma)} \\
&=\langle \partial_j f,g^{jk}\partial_k f\rangle_{L^2(S)} + \langle f,(K-M^2)f\rangle_{L^2(S)}.
\end{align*}
Comparing this to~\eqref{eq-lsd} we conclude that $Q=K_1^*I^*J^* L_S^D J I K_1$.

\subsection{Convergence of eigenvalues}
\label{s: convergence}
In view of Eqs. \eqref{eq-ccc}, \eqref{eq-qc2}, and \eqref{eq-qc1}, the operators $F_\beta:=E-(\widehat C_\beta+E)^{-1}$ form a monotonically increasing family converging strongly to the operator $G:=E-K_1^*I^*J^* (L_S^D+E)^{-1} J I K_1$ as $\beta\to+\infty$. Furthermore, for any $\varepsilon>0$,
all these operators have by~\eqref{eq-bc3} no essential spectrum in $(-\infty,E-\varepsilon)$ if $\beta$ is sufficiently large,
and by \cite{W2}, we have for any fixed $j$ with $\Lambda_j(G)<E-\varepsilon$ the convergence
\begin{equation}
     \label{eq-conv}
\Lambda_j(F_\beta)\to \Lambda_j(G) \;\;\text{ as } \beta\to+\infty\,.
\end{equation}
On the other hand, for any  fixed $j$ we can find $\varepsilon>0$ with
$\Lambda_j(Q)=\Lambda_j(L^D_S)=\mu^D_j < \varepsilon^{-1}-E$, and
\begin{equation}
    \label{eq-gf}
\Lambda_j(G)=E-( \mu_j^D+E)^{-1}<E-\varepsilon\,, \quad
\Lambda_j(F_\beta)=E-\big(\Lambda_j(\widehat C_\beta)+E)^{-1}\,,
\end{equation}
which means that the convergence~\eqref{eq-conv} holds for any fixed $j$. Substituting \eqref{eq-gf} into~\eqref{eq-conv} we deduce $\Lambda_j(\widehat C_\beta)=\mu_j^D+o(1)$ for $\beta\to+\infty$, and \eqref{eq-bc2} gives then $\Lambda_j(B^N_1)\ge -\beta^2/4+ \mu_j^D+o(1)$.
In combination with \eqref{eq-lll} and Lemma~\ref{lem1}, this concludes the proof of Proposition \ref{thm1}.


\section{Eigenvalue asymptotics with an improved remainder}

\subsection{Scheme of the proof}
\label{s: kernel}

Let us turn to the improved remainder estimate in Theorem \ref{thmain}. So far we have proved
$$
-\frac{\beta^2}{4} + \mu^D_j + o(1) \leq E_j(\beta) \leq -\frac{\beta^2}{4} + \mu^D_j +  \cO(\beta^{-1}\log \beta)
$$
holds as $\beta\to +\infty$. Now we are going to improve the lower bound to the same order of error term, $\cO(\beta^{-1}\log \beta)$, as the upper one
provided we adopt stronger assumptions on the regularity of the boundary. The estimates will closely follow the procedure used in \cite{EP} for an interaction supported by finite planar curves.
We provide first the main steps of the proof, and the technical details will be presented
in separate subsections below.

Our first aim is to estimate the decay of the eigenfunctions of $H$ with respect to the distance from $S$.
For this purpose, we will use an integral representation of the eigenfunction, which was obtained
in Corollary 2.3 of \cite{BEKS}: if $\lambda<0$ and $u\in \ker(H-\lambda)$, then
one can represent
\begin{align}\label{eq_gammafield}
u(x) = \iint_S \frac{ e^{-\sqrt{|\lambda|}\cdot|x-s|} }{4\pi|x-s|}\, h(s) \;d\sigma(s),
\end{align}
where $h\in L^2(S)$ is a solution to the integral equation
\begin{equation}\label{eq-hbb}
h(t) = \beta \iint_S \frac{ e^{-\sqrt{|\lambda|}\cdot|t-s|} }{4\pi|t-s|}\, h(s) \;d\sigma(s).
\end{equation}
In the rest of the section we are going to establish some relations between the above functions $u$ and $h$.
We remark that if $u$ and $h$ are related by \eqref{eq_gammafield}, then,
for almost every $x\in\R^3$, there holds
\begin{equation}
     \label{eq-udu}
     \begin{aligned}
\big|u(x)\big| &\leq  \frac{ e^{-\sqrt{|\lambda|}\dist(x,S)} }{4\pi\dist(x,S)}  \|h\|_{L^1(S)},\\
\big|\nabla u(x)\big| &\leq  \frac{1}{4\pi}\left( \sqrt{|\lambda|}+\frac{1}{\dist(x,S)}\right) \frac{ e^{-\sqrt{|\lambda|}\dist(x,S)} }{\dist(x,S)} \|h\|_{L^1(S)}
\end{aligned}
\end{equation}
where $\dist$ is the usual distance in $\RR^3$.
Therefore, in order to estimate the decay of the eigenfunctions $u$ of $H$, it is sufficient to
have a suitable bound for the norm of the associated functions $h$.

For sufficiently small $\delta$ the map $F$ given by \eqref{eq-sff} is a diffeomorphism between $S\times(-\delta,\delta)$ and
\begin{equation}
     \label{eq-box}
\boxminus_\delta := \big\{
s+t \nu(s): \, s\in S, \, t\in(-\delta,\delta)
\big\} .
\end{equation}
Hence, for a given function $u$ defined on $S$ we define the function $u_0$
on $\boxminus_\delta$ via
$$
u_0\big(s + t \nu(s)\big) := u( s).
$$
The following assertion is of crucial importance for the subsequent analysis, and it
is the most technically demanding part of this paper. Its proof by direct calculations is given in Subsection~\ref{proofuu0} below.
\begin{lemma}\label{u-u0}
There exist positive constants $C_1$, $C_2$ and $\delta_0$ with the following properties:
if  $\delta\in (0,\delta_0)$ and $\lambda< -C_1 (\log \delta)^2$ and $u$ and $h$ are related by
\eqref{eq_gammafield}, then
$\|u_0-u\|_{L^2(\boxminus_\delta)}
\leq C_2 \|h\|_{L^2(S)}  \cdot \delta$.
\end{lemma}

Assume now that $u$ is a normalized eigenfunction of $H$, then
Lemma \ref{u-u0} is used to obtain an upper bound for the norm of $h$, which gives then a pointwise upper bound
for $u$ using \eqref{eq-udu}.

\begin{lemma} \label{h_estim}
Let $j\in \NN$ be fixed and $u$ be a normalized eigenfunction of $H$
for the eigenvalue $E_j(\beta)$. Let $h$ be associated with $u$ by \eqref{eq_gammafield}, then
$\|h\|_{L^2(S)} =\cO(\beta^2)$ for large $\beta$.
\end{lemma}
\begin{proof} We have
\begin{equation}
      \label{equvu}
1=\|u\|_{L^2(\RR^3)} \geq \|u\|_{L^2(\boxminus_\delta)} \geq \|u_0\|_{L^2(\boxminus_\delta)}
-\|u-u_0\|_{L^2(\boxminus_\delta)}.
\end{equation}
Using the equality $h=\beta u|_S$, see \eqref{eq_gammafield} and \eqref{eq-hbb}, we get for $\delta$ small enough
\[
\|u_0\|_{L^2(\boxminus_\delta)}^2 \geq
\dfrac{1}{2} \int_{-\delta}^{\delta} \iint_S\big|u(s)\big|^2 d\sigma(s) \,dt = \delta \|u\|^2_{L^2(S)}= \delta \beta^{-2} \|h\|_{L^2(S)}^2.
\]
Now we take $\delta:=(A\beta)^{-2}$ with $A>0$.
By Lemmata \ref{lem1} and \ref{u-u0}, for sufficiently large $\beta$ we have $\|u_0-u\|_{L^2(\boxminus_\delta)}\le C_2 \|h\|_{L^2(S)}\delta$.
Therefore equation~\eqref{equvu} reads as
\[
1 \geq \Big(\beta^{-1}\sqrt{\delta} - C_2 \delta\Big) \|h\|_{L^2(S)}=\dfrac{1}{A\beta^2}\Big(
1-\dfrac{C_2}{A}\Big)\|h\|_{L^2(S)}.
\]
and the choice $A=2C_2$ gives $\|h\|_{L^2(S)}\le 4C_2 \beta^2$.
\end{proof}

The result of Lemma~\ref{h_estim} will be now used to obtain a new two-side estimate for the eigenvalues.
First we define one more auxiliary operator as follows. For $\varepsilon>0$ consider the set
\begin{equation}
\label{Theta_d}
\Theta_\varepsilon = \{z\in\R^3: \, \dist(z,S)<\varepsilon\}
\end{equation}
and the self-adjoint operator $K_\varepsilon$ acting in $L^2(\Theta_\varepsilon)$ generated by the quadratic form
\begin{equation}
\label{K_form}
k_\varepsilon(u,u)=\iiint_{\Theta_\varepsilon} |\nabla u|^2\, dx - \beta \iint_S|u|^2\, d\sigma
\,,\quad \cD(k_\varepsilon)=H_0^1(\Theta_\varepsilon) \,.
\end{equation}
Since the eigenfunctions of $H$ decay fast with the distance from $S$, one conjectures that the eigenvalues
of $H$ are close to those of $K_\varepsilon$ in a suitable asymptotic regime.
Our aim is now to put this guess on a more solid ground. Let us start with a technical preliminary.
In what follows we denote
\begin{equation}
   \label{eq-dk}
d:=\frac{k\log\beta}{\beta}, \quad
\delta:= d+\frac{1}{\beta} = \frac{1+k\log\beta}{\beta},
\end{equation}
where $k>1$ is a constant that will be chosen later. Moreover, let $\gamma\in C^\infty(\R)$ be such that $\gamma(s)=1$ for $s\geq 1$, $\gamma(s)\in (0,1)$ for $s\in (0,1)$, and $\gamma(s)=0$ for $s\leq 0$. For large $\beta$  set
\begin{align*}
g_\beta(x) :=
\begin{cases}
\gamma\left( \dfrac{  \log\big(\dist\big(x,\partial \Theta_\delta \big)\big)  + \log\beta }{ \log(\log\beta) } + 1  \right)  &\text{for } x\in \Theta_\delta\,,\\[.3em]
0	&\text{for } x\notin \Theta_\delta\,.
\end{cases}
\end{align*}
The function $g_\beta$ is absolutely continuous and its gradient $\nabla g_\beta(x) = \big(\partial_1 g_\beta(x), \partial_2 g_\beta(x), \partial_3 g_\beta(x)\big)$ exists for almost every $x$.

 The following result is obvious:

\begin{lemma}\label{lem:SuppNabla}
For large $\beta$, the support of $\nabla g_\beta$ is contained in the set
$$
\Omega(\beta) :=
\left\{ x\in \Theta_\delta : \frac{1}{\beta\log\beta} \leq \dist\big(x,\partial \Theta_\delta \big)\leq \frac{1}{\beta}  \right\},
$$
and $\Omega(\beta) \cap \Theta_d = \emptyset$.
\end{lemma}
Next let us fix an arbitrary $N\ge 1$ and consider an orthonormal family $(u_{j,\beta})_{j=1,\dots,N}$ of eigenfunctions
of $H$ for the eigenvalues $E_j(\beta)$, $j=1,\dots, N$. Introduce their cut-offs
\[
\phi_{j,\beta}:= g_\beta u_{j,\beta}.
\]
Remark that $g_\beta \in H_0^1(\Theta_\delta)$ and that  $\,g_\beta$  and $\,\nabla g_\beta$ are bounded by (\ref{dgbound}) below and Lemma \ref{lem:SuppNabla},
which gives the inclusions $\phi_{j,\beta} \in H^1_0(\Theta_\delta)$.

The proofs of the following two Lemmata are given in Subsections \ref{prooflem16} and \ref{lemonb} respectively.

\begin{lemma}\label{lem:Lemma16}
For any $N\in \NN$ and $\mu>0$ there exists $\kappa_1>0$ such that
for any $k\ge \kappa_1$ in \eqref{eq-dk}
and any $j,l=1,\dots,N$ there holds
$\left|  k_\delta(\phi_{j,\beta},\phi_{l,\beta})  - E_j(\beta)  \delta_{j,l}  \right| \leq \beta^{-\mu}$
as $\beta$ is large.
\end{lemma}

\begin{lemma}\label{lem:QuasiONB}
For any $N\in \NN$ and $\mu>0$ there exists $\kappa_2>0$
such that for any $k>\kappa_2$ in \eqref{eq-dk} and any $j,l=1,\dots,N$ there holds
$\big| \langle{  \phi_{j,\beta}  ,  \phi_{l,\beta}  }\rangle_{L^2(\Theta_\delta)}  - \delta_{j,l}  \big|
\leq \beta^{-\mu}$ as $\beta$ is large.
\end{lemma}

Now we can compare the eigenvalues of $H$ with those of $K_\delta$ as follows:
\begin{lemma}
\label{Kasympt}
Let $j\in\NN$ be fixed, then for any $\mu>0$ there exists $\kappa_0>0$ such that for any $k>\kappa_0$ in \eqref{eq-dk}
there holds
$\Lambda_j(K_\delta) - \beta^{-\mu} \leq  E_j(\beta)   \leq \Lambda_j(K_\delta)$
as $\beta$ is large.
\end{lemma}
\begin{proof}
By the min-max principle we have
\begin{equation}
   \label{eq-minmax}
\Lambda_j(K_\delta)
= \min_{\genfrac{}{}{0pt}{1}{  U\subset H_0^1(\Theta_\delta)  }{\dim U=j} }  \max_{ u\in U\setminus\{0\}  }  \frac{ k_{\delta}(u,u) }{ \|u\|^2_{L^2(\Theta_\delta)} }.
\end{equation}
As the extension of any function from $H_0^1(\Theta_\delta)$ by zero belongs to the form domain of $H$, we have immediately
$E_j(\beta)   \leq \Lambda_j(K_\delta)$.

Let $U$ be the  subspace spanned by the $\phi_{n,\beta}$ with $n=1,\ldots,j$.
By Lemma~\ref{lem:QuasiONB} one has $\dim U=j$ for large $\beta$. Let $b_1 \phi_{1,\beta}+\dots+b_j \phi_{j,\beta}=:\phi\in U$,
$b=(b_1,\dots,b_j)\in \CC^j$. Due to Lemmata~\ref{lem:Lemma16} and~\ref{lem:QuasiONB} one can find $k_0>0$
such that for $k\ge k_0$ we have for sufficently large $\beta$
\begin{align*}
\big(1- C \beta^{-\mu}\big)\|b\|^2_{\CC^j} \leq &\|\phi\|^2_{L^2(\Theta_\delta)} \leq \big(1+ C \beta^{-\mu}\big)\|b\|^2_{\CC^j}, \\
k_\delta(\phi,\phi)\le \sum_{n=1}^j
E_n(\beta) |b_n|^2 &+ C \beta^{-\mu}\|b\|^2_{\CC^j}
\le \Big(E_j(\beta)  + C \beta^{-\mu} \Big) \|b\|^2_{\CC^j},
\end{align*}
where $C>0$ is independent of $b$ and $\beta$. Using $\frac{1}{2}<1- C \beta^{-\mu}<1$ and $E_j(\beta)<0$ for $\beta$ sufficently large we get
\begin{align*}
\frac{ k_\delta(\phi,\phi) }{ \|\phi\|^2_{L^2(\Theta_\delta)} }
&\leq
\frac{ E_j(\beta) \|b\|^2_{\CC^j} }{ \|\phi\|^2_{L^2(\Theta_\delta)} }
+
\frac{ C\beta^{-\mu} \|b\|^2_{\CC^j} }{ \|\phi\|^2_{L^2(\Theta_\delta)} }
\leq
\frac{ E_j(\beta) \|b\|^2_{\CC^j} }{ \big(1+ C \beta^{-\mu}\big)\|b\|^2_{\CC^j} }
+
\frac{ C\beta^{-\mu} \|b\|^2_{\CC^j} }{ \big(1- C \beta^{-\mu}\big)\|b\|^2_{\CC^j} }  \\
&=
E_j(\beta) \left( 1 - \frac{ C \beta^{-\mu} }{ 1+ C \beta^{-\mu} } \right)
+
2 C\beta^{-\mu}
\leq
E_j(\beta)  -  C \beta^{-\mu}E_j(\beta) +  2 C\beta^{-\mu}  .
\end{align*}
Testing on the subspace $U$ in \eqref{eq-minmax}
and using $E_j(\beta)=\cO(\beta^2)$ we obtain $\Lambda_j(K_\delta)\le E_j(\beta)+C_1 \beta^{2-\mu}$,
where $C_1>0$ is independent of $\beta$. As $\mu>0$ can be chosen arbitrary, the result follows.
\end{proof}

\begin{lemma}\label{thm2}
Let the surface $S$ be with a compact $C^2$-boundary,
then for each fixed $j$ there holds
$E_j(\beta) \geq -\beta^2/4 + \mu_j^D + \cO(\beta^{-1} \log\beta)$
as $\beta\to+\infty$.
\end{lemma}
\begin{proof}
For sufficiently small $a>0$, consider the surface
$S_a=\{z\in S_1:\, \dist_{S_1}(z,S)<a\}$,
the distance $\dist_{S_1}$ being measured along the geodesics of $S_1$, and
let $\mu_j^D(a)$ be the $j$th Dirichlet eigenvalue of  $-\Delta_{S_a} + K - M^2$ on $S_a$.
It is a standard result of the domain perturbation theory,
see e.g. \cite{grinfeld}, that each $a\mapsto \mu^D_j(a)$ is Lipschitz for small $a$,
in particular, \begin{equation}
\label{mustability}
\mu_j^D(a) = \mu_j^D+\cO(a).
\end{equation}
Define
\[
\Xi_a := F\big(S_a\times(-a,a)\big) = \big\{s+\nu(s)t\; :\; s\in S_a,\, t\in(-a,a) \big\}
\]
and denote by $H^D_a$ the self-adjoint operator in $L^2(\Xi_a)$ generated
by the quadratic form
\[
h^D_a(u,u)=\iiint_{\Xi_a} |\nabla u|^2 dx -\beta \iint_{S_a} |u|^2 d\sigma,
\quad
\cD(h^D_a)=H^1_0(\Xi_a).
\]
Choose $k\ge 3$ to obtain the estimate of Lemma~\ref{Kasympt} with $\mu=1$
and set
\[
\delta:=\dfrac{1+k\log \beta}{\beta}, \quad a:= \dfrac{2k \log \beta}{\beta},
\]
For large $\beta$ we have the inclusion $\Theta_\delta\subset \Xi_a$ and
the inequalities $\Lambda_j(H^D_a)\le\Lambda_j(K_\delta)$.
Due to the Lemma \ref{Kasympt} and our choice of $k$ we have
$\Lambda_j(K_\delta)\le E_j(\beta)+\beta^{-1}$.
Furthermore, the operator $H^D_a$ can be studied using the separation variables (see Remark~\ref{rem3}), which gives
\[
\Lambda_j(H^D_a)=-\beta^2/4+\mu^D_j(a)+\cO(\beta^{-1}\log \beta)
=-\beta^2/4 + \mu^D_j +\cO(\beta^{-1}\log \beta),
\]
where we used \eqref{mustability}. Putting together the three inequalities for the eigenvalues
one arrives at
\[
-\beta^2/4 + \mu^D_j +\cO(\beta^{-1}\log \beta)
=\Lambda_j(H^D_a)\le \Lambda_j(K_\delta)
\le E_j(\beta)+\beta^{-1},
\]
which gives the sought estimate $E_j(\beta)\ge -\beta^2/4 + \mu_j^D + \cO(\beta^{-1}\log \beta)$.
\end{proof}

Theorem~\ref{thmain} is now completely contained in Lemma~\ref{lem1}, Proposition~\ref{thm1}
and Lemma~\ref{thm2}.

\subsection{Proof of Lemma \ref{u-u0}}\label{proofuu0}

We assume first that the surface $S_1$ can be covered by a single map $\Phi:D_1\to S_1\subset\R^3$, where
$D_1\subset\R^2$ is an open set and $\Phi$ is $C^4$. Denote
$D:=\Phi^{-1}(S)$. Remark that then there exists a bounded function $R:D\to\R^3$ such that
\begin{equation}      \label{eq-RR}
\Phi(z) = \Phi(y) + \Phi'(y)(z-y) + R(z,y)|z-y|^2
\end{equation}
holds for $z,y\in D$, where $\Phi'$ is the differential of $\Phi$.
Recall that
\[
d\sigma(s)=\left| \dfrac{\partial\Phi}{\partial y_1}\times\dfrac{\partial\Phi}{\partial y_2} \right|  \,dy_1 dy_2\equiv\sqrt{g(y)} dy_1 dy_2 \quad
\text{ for } s=\Phi(y_1,y_2).
\]
Furthermore, the map
$D\times(-\delta,\delta)\ni(y,t)\mapsto \Psi(y,t):= \Phi(y)+t \nu\big(\Phi(y)\big)\in\boxminus_\delta$
is a diffeomorphism, and  one can find $m>0$ such that
\begin{equation}
     \label{eq-phim}
\big|\Psi(y,t) - \Psi(z,s)\big|^2\geq m^2 \Big( |y-z|^2 +|s-t|^2\Big) \text{ for } (y,t),(z,s)\in D\times(-\delta,\delta).
\end{equation}
For $x_0\in S$, set $y := \Phi^{-1}(x_0)\in D$ and $x:=\Psi(y,t_0)\in \boxminus_\delta$ with $|t_0|<\delta$. We have
\begin{multline}
u\circ\Psi (y,0) - u\circ \Psi (y,t_0)
= -\int_0^{t_0} \partial_t \big( u\circ\Psi \big)(y,t) \; dt
= -\int_0^{t_0}  \Big\langle \nabla u\big(\Psi(y,t)\big), \nu(x_0)\Big\rangle  \; dt .
\end{multline}
Using~\eqref{eq_gammafield} we estimate
\[
\Big| \Big\langle{ \nabla u(x) , \nu(x_0) }\Big\rangle\Big|  \leq \frac{1}{4\pi} \iint_S \left[  \sqrt{|\lambda|} + \frac{1}{|x-s|}  \right] \frac{ e^{-\sqrt{|\lambda|}|x-s|} }{|x-s|^2} \cdot \big|h(s)\big|
  \cdot \Big|\Big\langle{ x-s ,\nu(x_0) }\Big\rangle\Big|  \;ds,
\]
and with the help of \eqref{eq-phim} one estimates, for $z\in D$,
\[
\big|\Psi(y,t)-\Phi(z)\big|
=\big|\Psi(y,t)-\Psi(z,0)\big| \geq \dfrac{m}{\sqrt 2} \big(|y-z|+|t|\big)
\]
and denoting $\omega(z,y) := - \big\langle R(z,y) , \nu(x_0) \big\rangle$ with $R$ from \eqref{eq-RR} we get
\begin{align*}
\big\langle \Psi(y,t)-\Phi(z) , \nu(x_0) \big\rangle &
= \big\langle \Phi(y)+t\nu(x_0) - \Phi(z), \nu(x_0) \big\rangle
=  \big\langle \Phi(y)- \Phi(z) , \nu(x_0) \big\rangle + t  \\
& =  \big\langle  -\Phi'(y)(z-y) - R(z,y)|z-y|^2 , \nu(x_0) \big\rangle + t
=  \omega(z,y)|z-y|^2 + t,
\end{align*}
and we note that  $\|\omega\|_\infty\leq \|R\|_\infty$. Hence,
\begin{align*}
\Big| \Big\langle\nabla u(x), \nu(x_0) \Big\rangle\Big|
& \leq
\frac{1}{4\pi} \iint_D \left[  \sqrt{|\lambda|} + \frac{1}{\frac{m}{\sqrt{2}}\big( |y-z|+|t| \big)} \right]
\frac{ e^{-\sqrt{|\lambda|}\frac{m}{\sqrt{2}}\big( |y-z|+|t|  \big)} }{m^2(|y-z|^2+|t|^2)} \\
& \hspace{10em}
\cdot\Big|h\big(\Phi(z)\big)\Big|  \cdot \Big|t+\omega(z,y)|z-y|^2\Big| \cdot
 \sqrt{g(z)}   \;dz \;.
\end{align*}
Denote $\mu:= m\sqrt{|\lambda|/2}$, $\,c_r:=\max\big\{1, \|R\|_\infty, \frac{\sqrt{2}}{m} \big\}^2$,
and introduce $\tilde{h}:\R^2\to\R$ by
\begin{align*}
\tilde{h}(z) &:=
\begin{cases}
\big|h\big(\Phi(z)\big)\big| \cdot \sqrt{g(z)} &\text{for } z\in D\,,\\
0	&\text{for } z\in \R^2\setminus D\,.
\end{cases}
\end{align*}
Then the preceding inequality takes the form
\[
\Big| \Big\langle \nabla u(x) , \nu(x_0) \Big\rangle\Big|
\leq
\frac{c_r}{4\pi m^2}   \,e^{-\mu|t|}  \iint_{\RR^2} \left[  \sqrt{|\lambda|} + \frac{1}{|y-z|+|t|}  \right]
\frac{ |y-z|^2+|t| }{|y-z|^2+|t|^2} \cdot e^{-\mu|y-z|} \cdot \widetilde{h}(z) \;dz,
\]
and using the functions $F_t, G_t:\R^2\to\R$,
$$
F_t(z) := \frac{ |z|^2+|t| }{|z|^2+|t|^2} \cdot e^{-\mu|z|}, \quad
G_t(z) := \frac{ 1 }{|z|+|t|} \cdot F_t(z)\,,
$$
it can be rewritten as
$$
\Big| \langle{ \nabla u(x) , \nu(x_0) }\Big\rangle \Big|
\leq
\frac{c_r}{4\pi m^2}\,   e^{-\mu|t|}
\Big[ \sqrt{|\lambda|}  \big(F_t \star\tilde{h} \big)(y) + \big(G_t \star \tilde{h} \big)(y) \Big]\,,
$$
where $\star$ means the convolution in $L^2(\R^2)$.
Denoting $c_d := \sqrt{\|\det \Psi'\|_\infty}$ and combining the preceding estimates
we then obtain
\begin{equation}
\begin{aligned}
\frac{1}{c_d^2} \|u_0-u\|_{L^2(\boxminus_\delta)}^2
&= \frac{1}{c_d^2}\big\|[u_0\circ\Psi -u\circ\Psi  ] \cdot \sqrt{|\det \Psi'|} \,\big\|_{L^2(D\times (-\delta,\delta))}^2  \\
& \leq\int_{-\delta}^\delta  \iint_{D}    \big|u\circ\Psi(y,0)-u\circ\Psi(y,t_0)\big|^2  dy \,dt_0  \\
&=\int_{-\delta}^\delta  \iint_{D} \left| \int_0^{t_0} \Big|\Big\langle \nabla u\big(\Psi(y,t)\big), \nu(x_0)\Big\rangle\Big|  \; dt  \right|^2  \;dy \,dt_0  \\
&\leq\int_{-\delta}^\delta |t_0| \int_0^{|t_0|} \iint_{D} \Big| \Big \langle \nabla u(\Psi(y, t\,{\rm sgn}(t_0)) , \nu(x_0)\Big\rangle\Big|^2   \;dy \;dt \;dt_0  \\
&\leq 2 \left(\frac{c_r}{4\pi m^2} \right)^2
\int_0^\delta t_0 \int_0^{t_0} e^{-2\mu|t|}  \iint_{\R^2}
\Big| \sqrt{|\lambda|}  \big(F_t\star\tilde{h} \big)(y) +  \big(G_t\star\tilde{h} \big)(y)  \Big|^2
\,dy \,dt \,dt_0  \\
&\leq 2 \left(\frac{c_r}{4\pi m^2} \right)^2
\int_0^\delta t_0 \int_0^{t_0} e^{-2\mu t}
\Big( \sqrt{|\lambda|}\,\|F_t\star \tilde{h}\|_{L^2(\R^2)}   + \|G_t\star \tilde{h}\|_{L^2(\R^2)}   \Big)^2
\,dt \,dt_0 .
\end{aligned}
      \label{eq-qqq}
\end{equation}
We have $\dfrac{r^2+|t|}{r^2+|t|^2}\,r \leq r+1$, and the Young inequality gives
\begin{multline*}
\|F_t\star\tilde{h}\|_{L^2(\R^2)} \leq \|F_t\|_{L^1(\R^2)} \|\tilde{h}\|_{L^2(\R^2)}
\leq \|{g}^{\frac{1}{4}}\|_\infty  \|h\|_{L^2(S)} \|F_t\|_{L^1(\R^2)}\\
 \leq
2\pi  \|{g}^{\frac{1}{4}}\|_\infty\|h\|_{L^2(S)} \int_0^\infty \frac{ r^2+|t| }{r^2+|t|^2} \cdot e^{-\mu r} r \;dr
\leq
 2\pi  \|{g}^{\frac{1}{4}}\|_\infty\|h\|_{L^2(S)} \int_0^\infty (r+1) e^{-\mu r} \;dr .
\end{multline*}
If $\lambda\leq -\frac{8}{m^2}$ we have $\mu\geq2$ and therefore $\sqrt{|\lambda|}\frac{\mu+1}{\mu^2} = \frac{\sqrt{2}\mu}{m} \cdot \frac{\mu+1}{\mu^2}  < \frac{3}{m} \leq \frac{3e^{\mu|t|}}{m}$ for all $t\in\R$. Hence
\begin{align*}
\|F_t\star\tilde{h}\|_{L^2(\R^2)}
\leq
2\pi  \frac{\mu+1}{\mu^2} \|{g}^{\frac{1}{4}}\|_\infty \|h\|_{L^2(S)}
\leq
2\pi \|{g}^{\frac{1}{4}}\|_\infty \|h\|_{L^2(S)} e^{\mu|t|} \frac{3}{m\sqrt{|\lambda|}}.
\end{align*}
Furthermore,
\begin{align*}
\|G_t\star\tilde{h}\|_{L^2(\R^2)} & \leq
\|G_t\|_{L^1(\R^2)} \|\tilde{h}\|_{L^2(\R^2)}
\leq 2\pi  \|{g}^{\frac{1}{4}}\|_\infty \|h\|_{L^2(S)} \int_0^\infty \frac{1}{r+|t|} \cdot \frac{ r^2+|t| }{r^2+|t|^2} \cdot e^{-\mu r} r \;dr .
\end{align*}
As $\frac{1}{r+|t|} \cdot \frac{ r^2+|t| }{r^2+|t|^2} r \leq \frac{r+1}{r+|t|} \leq \frac{r+|t|+1}{r+|t|}$ we have for $\mu\geq2$
\begin{align*}
\int_0^\infty \frac{1}{r+|t|} \cdot \frac{ r^2+|t| }{r^2+|t|^2} \cdot e^{-\mu r} r \;dr
\leq
\int_{|t|}^\infty \frac{r+1}{r}  e^{-\mu (r-|t|)} \;dr
=
e^{\mu|t|}\left(  \int_{|t|}^\infty e^{-\mu r}  \;dr + \int_{|t|}^\infty \frac{ e^{-\mu r} }{r}   \;dr\right) \\
\leq
e^{\mu|t|}  \left(  \frac{ e^{-\mu |t|} }{\mu} +
\int_{|t|}^1 \frac{ 1 }{r}   \;dr  +  \int_{1}^\infty  e^{-\mu r}   \;dr  \right)
=  e^{\mu|t|}  \left(  \frac{ e^{-\mu |t|} }{\mu} - \log|t|  +  \frac{ e^{-\mu} }{\mu}  \right)
\leq  e^{\mu|t|}  (1-\log|t|).
\end{align*}
This leads to
\begin{align*}
\sqrt{|\lambda|} \|F_t\star\widetilde{h}\|_{L^2(\R^2)} + \|G_t\star\widetilde{h}\|_{L^2(\R^2)}
\leq
2\pi \|{g}^{\frac{1}{4}}\|_\infty \|h\|_{L^2(S)} e^{\mu |t|}  \left(\frac{3}{m} + 1-\log|t|\right) .
\end{align*}
The substitution into \eqref{eq-qqq} gives
\begin{align*}
\|u_0-u\|_{L^2(\boxminus_\delta)}^2 & \leq
\left(\frac{\sqrt{2}c_rc_d}{4\pi m^2} \right)^2
\int_0^\delta t_0 \int_0^{t_0} e^{-2\mu t} \left( 2\pi \|{g}^{\frac{1}{4}}\|_\infty \|h\|_{L^2(S)} e^{\mu t}
\left(\frac{3}{m}+1-\log t\right) \right)^2\, dt \, dt_0\\
& \leq
\left(\frac{\sqrt{2}c_rc_d}{2 m^2}\right)^2 \|\sqrt{g}\|_\infty \|h\|_{L^2(S)}^2
\int_0^\delta t_0 \int_0^{t_0} \big(\log t-\log c_m\big)^2\, dt \, dt_0
\end{align*}
with $c_m:= e^{\frac{3}{m}+1}$. Assuming now that $\delta\geq t_0$ is sufficiently small we have
\begin{align*}
\int_0^{t_0} \big(\log t-\log c_m\big)^2\, dt
=
c_m\int_0^{\frac{t_0}{c_m}} (\log s)^2\, ds
=
c_m \Big[ s(\log s -1)^2 +s  \Big]_0^{\frac{t_0}{c_m}}
\leq \sqrt{t_0}
\end{align*}
and
$$
\int_0^\delta t_0 \int_0^{t_0} \big(\log t-\log c_m\big)^2 \,dt\,dt_0\leq
\int_0^\delta t^{3/2}_0 \,dt_0=\dfrac{2}{5} \delta^{5/2}\le \delta^2.
$$
Hence there exist $\lambda_0<0$ and $\delta_0>0$ such that for $\lambda<\lambda_0$ and $\delta\in(0,\delta_0)$ we have
\[
\|u_0-u\|_{L^2(\boxminus_\delta)}\leq C \|h\|_{L^2(S)}  \delta,
\quad
C:=\left(\frac{\sqrt{2} c_r c_d}{2 m^2}\right) \|{g}^{\frac{1}{4}}\|_\infty.
\]

For the general case ($S_1$ is not covered by a single map) we represent $S_1=\bigcup_{\alpha=1}^N S^\alpha$, where
each $S^\alpha$ is covered by a single map $\Phi^\alpha:D^\alpha\to S^\alpha$. Furthermore, represent $h=h_1+\dots+ h_N$, where the supports
of $h_\alpha$
have only zero-measure intersections in $S$
and such that the support of each term
$h_\alpha$ is contained in $S^\alpha$. Without loss of generality we may assume that for sufficiently
small (but fixed) $\rho>0$ there holds $|s-x|>\rho$ for $s\in \supp h_\alpha$,
$x\in \boxminus_\delta \setminus F\big(S^\alpha\times(-\delta,\delta)\big)$
with any $\alpha$; here $F$ is the map given in~\eqref{eq-sff}.
Consider the functions $u^\alpha$ associated with $h_\alpha$ by \eqref{eq_gammafield}
and represent $u^\alpha=\chi_\alpha u^\alpha +(1-\chi_\alpha) u^\alpha$, where $\chi_\alpha\big(s+t\nu(s)\big)=1$
if $s\in S^\alpha$, otherwise $\chi_\alpha\big(s+t\nu(s)\big)=0$.
Now we pick an arbitrary $\alpha$. According to the first part of the proof we have $\|\chi_\alpha (u^\alpha_0-u^\alpha)\|_{L^2(\boxminus_\delta)}\leq C \|h_\alpha\|_{L^2(S)} \delta$ if $\beta$ is large and $\delta$ is small, with some $C>0$ common for all $\alpha$.
On the other hand, due to \eqref{eq-udu}, for any $x\in\boxminus_\delta\setminus\supp \chi_\alpha$ we have
\[
\big|u^\alpha(x)\big| \le \dfrac{e^{-\sqrt {-\lambda}\rho}}{4\pi \rho} \|h_\alpha\|_{L^1(S)}
\le \sqrt{\area S^\alpha} \dfrac{e^{-\sqrt {-\lambda}\rho}}{4\pi \rho} \|h_\alpha\|_{L^2(S)},
\]
which gives
\begin{multline*}
\big\|(1-\chi_\alpha)(u^\alpha_0-u^\alpha)\big\|_{L^2(\boxminus_\delta)}\le
\big\|(1-\chi_\alpha)u^\alpha_0\big\|_{L^2(\boxminus_\delta)}
+
\big\|(1-\chi_\alpha)u^\alpha\big\|_{L^2(\boxminus_\delta)}\\
\le
2 \sqrt{\area S^\alpha} \dfrac{e^{-\sqrt {-\lambda}\rho}}{4\pi \rho}
\sqrt {2\delta (\area S) + {\mathcal O}(\delta^2)}  \|h_\alpha\|_{L^2(S)}\le \dfrac{\area S}{\pi \rho} e^{-\sqrt {-\lambda}\rho}  \sqrt \delta\|h_\alpha\|_{L^2(S)}.
\end{multline*}
Taking $\lambda\le -C_1 (\log \delta)^2$ with a sufficiently large $C_1>0$, which can be chosen the same for all $\alpha$, gives
$\big\|(1-\chi_\alpha)(u^\alpha_0-u_\alpha)\big\|_{L^2(\boxminus_\delta)}\le \delta \|h_\alpha\|_{L^2(S)}$
and
\[
\|u^\alpha_0-u^\alpha\|_{L^2(\boxminus_\delta)}
\le
 \big\|\chi_\alpha(u^\alpha_0-u^\alpha)\big\|_{L^2(\boxminus_\delta)}
+
\big\|(1-\chi_\alpha)(u^\alpha_0-u^\alpha)\big\|_{L^2(\boxminus_\delta)}
\le (C+1) \delta \|h_\alpha\|_{L^2(S)}.
\]
Finally,
\[
\|u_0-u\|_{L^2(\boxminus_\delta)}\le \sum_{\alpha=1}^N \|u^\alpha_0-u^\alpha\|_{L^2(\boxminus_\delta)}
\le (C+1)\delta \sum_{\alpha=1}^N \|h_\alpha\|_{L^2(S)} = (C+1)\delta \|h\|_{L^2(S)},
\]
which gives the result with $C_2:=C+1$.\hfill$\square$

\subsection{Proof of Lemma \ref{lem:Lemma16}}\label{prooflem16}

Let us first give a preliminary estimate:
\begin{lemma}\label{lem:estim_nablag}
There are constants $C_1,C_2>0$ such that the following relations hold for large $\beta$:
$$
\iiint_{ \Omega(\beta) }   |\nabla g_\beta(x)|   \; dx     \leq C_1\,,  \quad
\iiint_{ \Omega(\beta) }   |\nabla g_\beta(x)|^2 \; dx     \leq C_2\, \beta \log\beta\,.
$$
\end{lemma}
\begin{proof}
For  $x\in \Theta_\delta$ we have
$$
\partial_{x_k} g_\beta(x) =
\gamma'\left( \frac{  \log\big(\dist\big(x,\partial \Theta_\delta \big)\big)  + \log\beta }{ \log(\log\beta) } + 1 \right) \cdot
\frac{ 1 }{ \log(\log\beta) }   \cdot \frac{  \partial_{x_k}  \dist\big(x,\partial \Theta_\delta \big)   }{ \dist\big(x,\partial \Theta_\delta \big) }\,.
$$
Since $|\partial_{x_k}  \dist\big(x,\partial \Theta_\delta \big)| \leq 1$ and $\|\gamma'\|_\infty <\infty$ one can find a constant $\tilde{C}>0$ such that
\begin{equation}
\label{dgbound}
|\nabla g_\beta(x)| \leq \frac{\tilde{C}}{\log(\log\beta)} \cdot \frac{1}{\dist\big(x,\partial \Theta_\delta \big)}\,,
\end{equation}
and consequently, for $\nu \in \{1,2\}$ we get
\begin{multline*}
\iiint_{ \Omega(\beta) }  |\nabla g_\beta(x)|^\nu   \;dx
\leq
\frac{\tilde{C}^\nu}{\big(\log(\log\beta)\big)^\nu} \cdot
\iiint_{ \Omega(\beta) }  \frac{1}{  |\dist\big(x,\partial \Theta_\delta \big)|^\nu  }   \; dx  \\
\leq \frac{\tilde{C}^\nu}{\big(\log(\log\beta)\big)^\nu} \cdot \left( \; C_5
\int_{ \frac{1}{\beta\log\beta} }^{ \frac{1}{\beta} }  \frac{1}{r^\nu} \; dr  \right)
\end{multline*}
with a certain constant $C_5>0$. Indeed, note that the normal variable $t$ is the distance of points $x\in \boxminus_\delta$ to $S$, and similarly for the points in tubular neighborhoods of the $C^2$ arcs forming $\partial S$ and the balls centered at their junctions. Evaluating the integrals,
$$
\int_{ \frac{1}{\beta\log\beta} }^{ \frac{1}{\beta} }  \frac{1}{r} \;dr
= \log(\log\beta)\,, \quad
\int_{ \frac{1}{\beta\log\beta} }^{ \frac{1}{\beta} } \frac{1}{r^2} \; dr
= -\beta+\beta\log\beta\,,
$$
we arrive at
\begin{equation*}
\iiint_{ \Omega(\beta) }  |\nabla g_\beta(x)|   \; dx  \leq  \tilde{C} C_5\,, \quad
\iiint_{ \Omega(\beta) }  |\nabla g_\beta(x)|^2   \; dx \leq
\tilde{C}^2 C_5 \frac{ \beta\log\beta-\beta }{\big(\log(\log\beta)\big)^2}
\leq
\tilde{C}^2 C_5 \beta\log\beta\,,
\end{equation*}
which proves the claim of the lemma.
\end{proof}

Now we are ready for the proof of Lemma~\ref{lem:Lemma16}. Remark that by Proposition~\ref{thm1}, for any constants $0<L_1<1/2<L_2$,
and all $j=1,\dots, N$ there holds
\begin{equation}
\label{L_estim}
-L_2^2 \beta^2 < E_j(\beta) < -L_1^2 \beta^2 \text{ for large $\beta$}.
\end{equation}

For the sake of brevity we set $\phi:=\phi_{j,\beta}$ and $\psi:=\phi_{l,\beta}$ as well as $u:=u_{j,\beta}$ and $v:=u_{l,\beta}$, and  moreover, $\lambda := E_j(\beta)$ and $\mu := E_l(\beta)$. Using
\[
\langle{\nabla u,\nabla v}\rangle_{L^2(\R^3)}    -\beta\iint_{S}{\overline{u}\cdot v}\,ds = E_j(\beta)\,\delta_{j,l},
\]
we get the identity
\begin{multline*}
\langle\nabla\phi_{j,\beta},\nabla\phi_{l,\beta}\rangle_{L^2(\R^3)}   -  \beta \iint_S \overline{\phi_{j,\beta}}\cdot \phi_{l,\beta} ds  - E_j(\beta)  \delta_{j,l} \\
=\langle{\nabla \phi,\nabla\psi}\rangle_{L^2(\Theta_\delta)} - \beta \iint_{S}{}{\overline{\phi}\cdot \psi}\;ds -
\left( \langle{\nabla u,\nabla v}\rangle_{L^2(\R^3)} - \beta \iint_{S}{}{\overline{u|_S}\cdot v|_S}\;ds \right)\\
= - \langle{\nabla u,\nabla v}\rangle_{L^2(\R^3)}  + \langle{\nabla \phi,\nabla\psi}\rangle_{L^2(\Theta_\delta)}\,.
\end{multline*}
We have
\begin{align*}
\langle{\nabla \phi,\nabla\psi}\rangle_{L^2(\Theta_\delta)}&
= \iiint_{\Theta_\delta}{}{ \overline{\Big(\nabla g_\beta(x)\cdot u(x) +  g_\beta(x)\cdot \nabla u(x)\Big)} \cdot
\Big( \nabla g_\beta(x)\cdot v(x) +  g_\beta(x)\cdot \nabla v(x)\Big)  }\;dx  \\
& = \iiint_{\Theta_\delta}{}{ |\nabla g_\beta(x)|^2\cdot \overline{u(x)} v(x)  }x +
\iiint_{\Theta_\delta}{}{ |g_\beta(x)|^2\cdot \Big(\overline{\nabla u(x)}\cdot \nabla v(x)} \Big)  \;dx \\
& \quad+ \iiint_{\Theta_\delta}{}{ \overline{u(x)}g_\beta(x)\Big(\nabla g_\beta(x)\cdot \nabla v(x)} \Big)  \;dx +
\iiint_{\Theta_\delta}{}{ g_\beta(x) v(x)\Big(\nabla \overline{u(x)} \cdot \nabla g_\beta(x) \Big)  }\;dx\,,
\end{align*}
hence,
\begin{multline*}
\langle{\nabla \phi,\nabla\psi}\rangle_{L^2(\Theta_\delta)} - \beta \iint_{S}{}{\overline{\phi}\cdot\psi}\;ds - E_j(\beta)\delta_{j,l}\\
\begin{aligned}
&= - \langle{\nabla u,\nabla v}\rangle_{L^2(\R^n\setminus \Theta_\delta)}
+\iiint_{\Theta_\delta}{}{ |\nabla g_\beta(x)|^2\cdot \overline{u(x)} v(x)  }\;dx\\
&\quad+ \iiint_{\Theta_\delta}{}{ \Big(|g_\beta(x)|^2-1\Big)\cdot \Big(\overline{\nabla u(x)} \cdot \nabla v(x) \Big)  }\;dx
+\iiint_{\Theta_\delta}{}{ \overline{u(x)}g_\beta(x)\Big(\nabla g_\beta(x)\cdot \nabla v(x) \Big)  }\;dx\\
&\quad + \iiint_{\Theta_\delta}{}{ g_\beta(x)v(x)\Big(\nabla \overline{u(x)} \cdot \nabla g_\beta(x) \Big)  }\;dx
=:I_1+I_2+I_3+I_4+I_5.
\end{aligned}
\end{multline*}
Denote by $h$ and $f$ the functions in $L^2(S)$ corresponding to $u$ and $v$ as in Eq.~\eqref{eq_gammafield}.
Using \eqref{eq-udu} and $\dist(x,S)\geq \delta$ for $x\notin\Theta_\delta$ we get with $\|h\|_1:=\|h\|_{L^1(S)}$
and $\alpha:=-(\sqrt{-\lambda}+\sqrt{-\mu})<0$
\begin{align*}
\left| \langle{\nabla u,\nabla v}\rangle_{L^2(\R^3\setminus \Theta_\delta)}  \right|&
\leq
\iiint_{\R^3\setminus \Theta_\delta}{}{ \big| \nabla u(x) \big| \cdot \big| \nabla v(x) \big| }\;dx \\
& \leq \frac{\|h\|_1 \|f\|_1}{16\pi^2 \delta^2} \left(\sqrt{-\lambda}+\frac{1}{\delta}\right)
 \left(\sqrt{-\mu}+\frac{1}{\delta}\right) \iiint_{\R^3\setminus \Theta_\delta}{}{ e^{\alpha\,\dist(x,S)}  }\;dx .
\end{align*}
The integral on the right-hand side can be estimated in the following way: we choose $R>1$ such that $\Theta_1$ is contained in the ball $B_R$ of radius $R$ around zero, then we have for $\delta\leq1$
\begin{multline*}
\lefteqn{\iiint_{\R^3 \setminus \Theta_\delta} e^{\alpha\, \dist(x,S)}\;dx =
\iiint_{\R^3 \setminus B_{2R}} e^{\alpha\, \dist(x,S)}\;dx + \iiint_{B_{2R} \setminus \Theta_\delta} e^{\alpha\, \dist(x,S)}\;dx} \\
\leq 4\pi \int_{2R}^\infty e^{\alpha(r-R)}r^2\;dr + e^{\alpha\delta}\frac{4}{3}\pi(2R)^3
 = 8\pi e^{\alpha R}\frac{1-2\alpha R +2\alpha^2 R^2}{-\alpha^3} + e^{\alpha\delta}\frac{32}{3}\pi R^3=:A.
\end{multline*}
For large $\beta$ we have due to \eqref{L_estim}
$$
-2L_2\beta \leq \alpha = -(\sqrt{-\lambda}+\sqrt{-\mu}) \leq -2 L_1 \beta
\qquad\text{and}\qquad
\frac{(1 + 2 L_2 \beta)^2}{(2 L_1 \beta)^3}\leq\frac{2}{3}.
$$
Since $\alpha$ is negative it follows
\begin{align*}
\frac{1-2\alpha R +2\alpha^2 R^2}{-\alpha^3}
\leq
\frac{2R^3(1-\alpha)^2}{-\alpha^3}
\leq 2R^3 \frac{(1 + 2 L_2 \beta)^2}{(2 L_1 \beta)^3}\leq\frac{4R^3}{3} .
\end{align*}
which implies $A\leq c(e^{\alpha R} + e^{\alpha\delta})$ with $c:= \frac{32}{3}\pi R^3$. Hence,
$$
|I_1|
\leq \frac{ \|h\|_1 \|f\|_1}{16\pi^2 \delta^2 } \left(\sqrt{-\lambda} + \frac{1}{\delta}\right)
\left(\sqrt{-\mu} + \frac{1}{\delta}\right) c \big(e^{\alpha R} + e^{\alpha\delta}\big)\,.
$$

Now note that $\supp\;\nabla g_\beta \subset \Omega(\beta)$ and $\Omega(\beta)\cap \Theta_d = \emptyset$, hence
\begin{multline*}
|  I_2|
\leq
\iiint_{\Theta_\delta}{}{ |\nabla g_\beta(x)|^2\cdot |u(x)|\cdot |v(x)|  }\;dx \\
\leq
\iiint_{\Omega(\beta)}{}{ |\nabla g_\beta(x)|^2\cdot
\frac{ e^{-\sqrt{-\lambda}\,\dist(x,S)} }{4\pi\dist(x,S)}\, \|h\|_1 \cdot
\frac{ e^{-\sqrt{-\mu}\,\dist(x,S)} }{4\pi\dist(x,S)}\, \|f\|_1   }\;dx \\
 \leq
\frac{\|h\|_1\|f\|_1}{16\pi^2}  \cdot \frac{ e^{\alpha d} }{d^2}
\iiint_{ \Omega(\beta) }{}{ |\nabla g_\beta(x)|^2  }\;dx.
\end{multline*}
As for $I_3$, we note that $g_\beta(x)=1$ iff $\dist(x,S)\leq d$, and moreover, $g_\beta(x)\leq 1$ for all $x$.
This gives
\begin{align*}
|I_3|&
\leq \iiint_{\Theta_\delta}{}{ \! \Big||g_\beta(x)|^2-1\Big|\cdot |\nabla u(x)|\cdot |\nabla v(x)|  }\;dx \\
& \leq
\iiint_{\Theta_\delta} \Big(1-|g_\beta(x)|^2\Big)\cdot
\frac{1}{4\pi}\left( \sqrt{-\lambda} + \frac{1}{\dist(x,S)}\right) \frac{ e^{-\sqrt{-\lambda}\,\dist(x,S)} }{\dist(x,S)} \,\|h\|_1 \\
& \qquad  \cdot
\frac{1}{4\pi}\left( \sqrt{-\mu}+\frac{1}{\dist(x,S)}\right) \frac{ e^{-\sqrt{-\mu}\,\dist(x,S)} }{\dist(x,S)}\, \|f\|_1 \;dx \\
& =
\frac{\|h\|_1\|f\|_1}{16\pi^2}
\iiint_{\Theta_\delta \setminus \Theta_d}  \Big(1-|g_\beta(x)|^2\Big)\cdot  \left( \sqrt{-\lambda}+\frac{1}{\dist(x,S)}\right) \\
& \qquad  \cdot
\left(\sqrt{-\mu}+\frac{1}{\dist(x,S)}\right) \frac{ e^{\alpha\, \dist(x,S)} }{\dist(x,S)^2}\;dx\\
& \leq
\frac{\|h\|_1\|f\|_1}{16\pi^2} \left( \sqrt{-\lambda}+\frac{1}{d}\right) \cdot
\left(\sqrt{-\mu}+\frac{1}{d}\right)
\frac{ e^{\alpha d} }{d^2}
\iiint_{\Theta_\delta\setminus \Theta_d} \big( 1-|g_\beta(x)|^2 \big)\;dx \,.
\end{align*}
Using $\supp\,\nabla g_\beta \cap \Theta_d = \emptyset$ and $|g_\beta(x)|\leq 1$ we obtain
\begin{align*}
|I_4|&
\leq \iiint_{ \Omega(\beta) }{}{ |u(x)| \cdot |\nabla g_\beta(x)| \cdot |\nabla v(x)|  }\;dx \\
& \leq
\iiint_{ \Omega(\beta) }
\frac{ e^{-\sqrt{-\lambda}\,\dist(x,S)} }{4\pi\dist(x,S)}\, \|h\|_1 \cdot |\nabla g_\beta(x)| \cdot
\frac{1}{4\pi}\left( \sqrt{-\mu}+\frac{1}{\dist(x,S)}\right) \frac{ e^{-\sqrt{-\mu}\,\dist(x,S)} }{\dist(x,S)}\, \|f\|_1  \;dx \\
& \leq
\frac{\|h\|_1\|f\|_1}{16\pi^2}  \left(\sqrt{-\mu}+\frac{1}{d}\right)
\frac{ e^{\alpha d}}{d^2}   \iiint_{ \Omega(\beta) }{}{  |\nabla g_\beta(x)|   }\;dx\,,
\end{align*}
and, similarly,
\[
|I_5|\leq \frac{\|h\|_1\|f\|_1}{16\pi^2}  \left(\sqrt{-\lambda}+\frac{1}{d}\right)
\frac{ e^{\alpha d} }{d^2}   \iiint_{ \Omega(\beta) }{}{  |\nabla g_\beta(x)|   }\;dx \,.
\]
Combining now all these estimates and realizing that $\delta>d>1/\beta$ for large $\beta$, so
$\sqrt{-\lambda}+\frac{1}{\delta}\leq \sqrt{-\lambda}+\frac{1}{d}\leq (L_2+1)\beta$ and analogically for $\mu$, we get
\begin{multline*}
B := \Bigg|
\langle{\nabla \phi,\nabla\psi}\rangle_{L^2(\Theta_\delta)} - \beta \iint_{S}{}{\overline{\phi}\cdot\psi}\;d\sigma -  \delta_{j,l} E_j(\beta)
\Bigg| \\
\begin{aligned}
& \leq
\frac{\|h\|_1\|f\|_1}{16\pi^2} \,(L_2+1)^2
\bigg[
\frac{\beta^2}{\delta^2}\,c (e^{\alpha R} + e^{\alpha \delta})
  +\,\frac{1}{(L_2+1)^2}\frac{e^{\alpha d}}{d^2} \iiint_{\Omega(\beta)} |\nabla g_\beta(x)|^2\; dx \\
& \quad +\,\beta^2\, \frac{e^{\alpha d}}{d^2} \iiint_{\Theta_\delta \setminus \Theta_d} (1-|g_\beta(x)|^2) \; dx
+ 2\beta\frac{e^{\alpha d}}{d^2} \iiint_{\Omega(\beta)} |\nabla g_\beta(x)|\;dx
\bigg].
\end{aligned}
\end{multline*}
For sufficiently large $\beta$ there holds
$$
\iiint_{\Theta_\delta \setminus \Theta_d} (1-|g_\beta(x)|^2) \leq \vol (\Theta_\delta) \leq 1 \,.
$$
Using further Lemma~\ref{lem:estim_nablag} we can proceed with the estimate
\begin{align*}
B\leq\,  \frac{\|h\|_1\|f\|_1}{16\pi^2}(L_2+1)^2 \cdot
\bigg[ \frac{c}{\delta^2}\beta^2\, \big(e^{\alpha R}+e^{\alpha\delta}\big)
 + \frac{1}{(L_2+1)^2}\,
\frac{e^{\alpha d}}{d^2}
\,C_2 \beta \log\beta  + \beta^2\, \frac{e^{\alpha d}}{d^2}
+ 2 \beta \frac{e^{\alpha d}}{d^2} \,C_1 \bigg] .
\end{align*}
Using Lemma~\ref{h_estim} and the Cauchy-Schwarz inequality we obtain $\|h\|_1= \mathcal{O}(\beta^2)$ and $\|f\|_1 = \mathcal{O}(\beta^2)$. With $\frac{1}{\delta}<\frac{1}{d}$, $e^{\alpha\delta}<e^{\alpha d}$ and $e^{\alpha(R-d)}<1$ we then get,
for sufficently large $\beta$,
$$
B
\leq
\frac{\|h\|_1\|f\|_1}{16\pi^2}(L_2+1)^2 \left[ 2c + \frac{C_2}{(L_2+1)^2} + 1 +2C_1 \right] \frac{e^{\alpha d}}{d^2}\,\beta^2
\leq
C \frac{e^{\alpha d}}{d^2}\,\beta^6
$$
for a certain $C>0$.
Substituting $\alpha:=-(\sqrt{-\lambda}+\sqrt{-\mu})\leq - 2 L_1 \beta$ and $d=k\beta^{-1}\log\beta$ we obtain
$$
\frac{e^{\alpha d}}{d^2}\beta^6 \leq k^{-2} \beta^{8-2L_1 k} (\log\beta)^{-2},
$$
and the result of Lemma ~\ref{lem:Lemma16} follows by taking $k$ sufficiently large.\hfill$\square$

\subsection{Proof of Lemma \ref{lem:QuasiONB}}\label{lemonb}

We know that $\langle{u_{j,\beta},  u_{l,\beta} }\rangle_{L^2(\R^3)} = \delta_{j,l}$, hence we can estimate the expression in question as
\begin{multline*}
A:=\Big| \langle{u_{j,\beta}, u_{l,\beta}  }\rangle_{L^2(\R^3)} - \langle{  \phi_{j,\beta},  \phi_{l,\beta}  }\rangle_{L^2(\Theta_\delta)}  \Big|
\\
 = \left|  \iiint_{\R^3}   \big(1-g_\beta(x)^2\big) \overline{u_{j,\beta}(x)}  u_{l,\beta}(x)  \;dx \right|
\leq
\iiint_{\R^3\setminus \Theta_d}   |\overline{u_{j,\beta}(x)}  u_{l,\beta}(x) | \;dx .
\end{multline*}
As in the previous proof, set $\|h\|_1:=\|h\|_{L^1(S)}$ and $\|f\|_1:=\|f\|_{L^1(S)}$.
In view of Eqs. \eqref{eq-udu} and \eqref{L_estim}, for large $\beta$
we have
\begin{align*}
A
&\leq
\iiint_{\R^3\setminus \Theta_d}
\frac{ e^{-\sqrt{-E_j(\beta)}\,\dist(x,S)} }{4\pi\,\dist(x,S)} \,\|h\|_1 \cdot
\frac{ e^{-\sqrt{-E_l(\beta)}\,\dist(x,S)} }{4\pi\,\dist(x,S)} \,\|f\|_1 \;dx \\
&\leq
\frac{  \|h\|_1\|f\|_1  }{16\pi^2 d^2}
\iiint_{\R^3\setminus \Theta_d}
e^{-\big( \sqrt{-E_j(\beta)}+\sqrt{-E_l(\beta)} \big) \dist(x,S)}  \;dx
\leq
\frac{  \|h\|_1\|f\|_1  }{16\pi^2 d^2}
\iiint_{\R^3\setminus \Theta_d}       e^{-2L_1\beta\,  \dist(x,S)}  \;dx\,.
\end{align*}
Mimicking the proof of Lemma~\ref{lem:Lemma16} one can check that
$$
\iiint_{\R^3 \setminus \Theta_d} e^{-2L_1\beta\, \dist(x,S)}\; dx \leq c \big(e^{-2L_1\beta R} + e^{-2L_1\beta d}\big)
$$
with an $R>1$ such that $\Theta_1$ is contained in the ball $B_R$ of radius $R$ centered at zero.
By Lemma~\ref{h_estim} we may estimate $\|h\|_{L^1(S)}\leq C\beta^2$ and $\|f\|_{L^1(S)}\leq C \beta^2$. As $d=\frac{k\log\beta}{\beta}$ we get
\begin{multline*}
| \langle{  \phi_{j,\beta}  ,  \phi_{l,\beta}  }\rangle_{L^2(\Theta_\delta)}  - \delta_{j,l}|
 \leq \frac{c}{16\pi^2 d^2}  \|h\|_{L^1(S)}\|f\|_{L^1(S)}  \, \big(e^{-2L_1\beta R} + e^{-2L_1\beta d}\big) \\
 \leq \frac{c C^2 \beta^6}{16\pi^2 k^2 (\log \beta)^2} \big(e^{-2L_1\beta R} + \beta^{-2L_1 k}\big),
 \end{multline*}
 and the result follows by taking $k$ sufficiently large.\hfill$\square$



\end{document}